\documentclass[a4paper,USenglish,numberwithinsect]{lipics-v2019}

\usepackage{import}

\newcommand{\set}[1]{\{#1\}}
\newcommand{\poseven}{\mathrm{Pos\-Even}}
\newcommand{\posodd}{\mathrm{Pos\-Odd}}
\newcommand{\limsupeven}{\mathrm{Limsup\-Even}}
\newcommand{\limsupodd}{\mathrm{Limsup\-Odd}}
\newcommand{\Aa}{\mathcal{A}}
\newcommand{\Ll}{\mathcal{L}}
\newcommand{\Rr}{\mathcal{R}}
\newcommand{\Ss}{\mathcal{S}}

\newcommand{\floor}[1]{\lfloor#1\rfloor}
\newcommand{\ceil}[1]{\lceil#1\rceil}
\newcommand{\seq}[1]{\langle #1 \rangle}
\newcommand{\rn}{\mathit{rn}}
\newcommand{\bad}{\mathit{bad}}
\newcommand{\fst}{\mathit{fst}}
\newcommand{\rej}{\mathsf{rej}}

\nolinenumbers
\hideLIPIcs

\title{Parity Games: Another View on Lehtinen's Algorithm}

\author{Paweł Parys}{Institute of Informatics, University of Warsaw, Poland}{parys@mimuw.edu.pl}{https://orcid.org/0000-0001-7247-1408}{}

\authorrunning{P. Parys}

\Copyright{Paweł Parys}



\funding{Work supported by the National Science Centre, Poland (grant no.\ 2016/22/E/ST6/00041).}

\begin{document}

\maketitle

\begin{abstract}
	Recently, five quasi-polynomial-time algorithms solving parity games were proposed.
	We elaborate on one of the algorithms, by Lehtinen (2018).
	
	Czerwiński et al.\ (2019) observe that four of the algorithms can be expressed as constructions of separating automata (of quasi-polynomial size),
	that is, automata that accept all plays decisively won by one of the players, and rejecting all plays decisively won by the other player.
	The separating automata corresponding to three of the algorithms are deterministic, and it is clear that deterministic separating automata can be used to solve parity games.
	The separating automaton corresponding to the algorithm of Lehtinen is nondeterministic, though.
	While this particular automaton  can be used to solve parity games, this is not true for every nondeterministic separating automaton.
	As a first (more conceptual) contribution, we specify when a nondeterministic separating automaton can be used to solve parity games.
	
	We also repeat the correctness proof of the Lehtinen's algorithm, using separating automata.
	In this part, we prove that her construction actually leads to a faster algorithm than originally claimed in her paper: 
	its complexity is $n^{O(\log n)}$ rather than $n^{O(\log d\cdot\log n)}$ (where $n$ is the number of nodes, and $d$ the number of priorities of a considered parity game),
	which is similar to complexities of the other quasi-polynomial-time algorithms.
\end{abstract}

\section{Introduction}

	Parity games have played a fundamental role in automata theory, logic, and their applications to verification and synthesis since early 1990's.
	The algorithmic problem of finding the winner in parity games can be seen as the algorithmic backend to problems in automated verification and controller synthesis.
	It is polynomial-time equivalent to the emptiness problem for nondeterministic automata on infinite trees with parity acceptance conditions, 
	and to the model-checking problem for modal $\mu$-calculus~\cite{EJS01}.
	Also, decision problems like validity or satisfiability for modal logics can be reduced to parity game solving.
	Moreover, it lies at the heart of algorithmic solutions to the Church's synthesis problem~\cite{RabinBook}. 
	The impact of parity games reaches relatively far areas of computer science, like Markov decision processes~\cite{FearnleyMDP} and linear programming~\cite{FHZ-simplex}.

	The problem of solving parity games has interesting complexity-theoretic status.
	It is a long-standing open question whether parity games can be solved in polynomial-time.
	Several results show that they belong to some classes ``slightly above'' polynomial time.
	Namely, deciding the winner of parity games was shown to be in $\mathsf{NP}\cap\mathsf{coNP}$~\cite{EJS01}, and in $\mathsf{UP}\cap\mathsf{coUP}$~\cite{up-co-up},
	while computing winning strategies is in \textsf{PLS}, \textsf{PPAD}, and even in their subclass \textsf{CLS}~\cite{Daskalakis-Papadimitriou}.
	The same holds for other kinds of games: mean-payoff games~\cite{mean-payoff}, discounted games, and simple stochastic games~\cite{stochastic};
	parity games, however, are the easiest among them, in the sense that there are polynomial-time reductions from parity games to the other kinds of games~\cite{up-co-up,mean-payoff},
	but no reductions in the opposite direction are known.
	
	For almost three decades researchers were trying to cutback the complexity of solving parity games, which resulted in a series of algorithms, 
	all of which were either exponential~\cite{Zielonka,BCJLM97,Seidl96,old-progress-measure,strategy-improvement,Schewe-big-steps,priority-promotion},
	or mildly subexponential~\cite{randomized-subexponential,subexponential}.
	The next era came unexpectedly in 2017 with a breakthrough result of Calude, Jain, Khoussainov, Li, and Stephan~\cite{calude} (see also~\cite{parity-short,parity-excursion}), 
	who designed an algorithm working in quasi-polynomial time (QPT for short).
	This invoked a series of QPT algorithms, which appeared soon after~\cite{small-progress-measure,Fearnley,Lehtinen,ZielonkaParys}.

	Four of the QPT algorithms~\cite{calude,small-progress-measure,Fearnley,Lehtinen}, at first glance being quite different,
	actually proceed along a similar line---as observed by Bojańczyk and Czerwiński~\cite[Section 3]{separation-toolbox} and Czerwiński et al.~\cite{lower-bound}.
	Namely, out of all the four algorithms one can extract a construction of a \emph{PG separator}, that is, a safety automaton 
	(nondeterministic in the case of Lehtinen~\cite{Lehtinen}, and deterministic in the other algorithms), 
	which accepts all words encoding plays that are decisively won by one of the players 
	(more precisely: plays consistent with some positional winning strategy), 
	and rejects all words encoding plays in which the player loses (for plays that are won by the player, but not decisively, the automaton can behave arbitrarily).
	The PG separator does not depend at all on the game graph; it depends only on its size.
	Having a PG separator, it is not difficult to convert the original parity game into an equivalent safety game
	(by taking a ``product'' of the parity game and the PG separator), which can be solved easily---and all the four algorithms actually proceed this way,
	even if it is not stated explicitly that a PG separator is constructed.
	As shown in Czerwiński et al.~\cite{lower-bound} (see also Colcombet and Fijalkow~\cite{universal-graphs} for another view on this proof), 
	all PG separators have to look very similar: their states have to be leaves of some so-called universal tree;
	particular papers propose different constructions of these trees, and of the resulting PG separators (of quasi-polynomial size).
	Moreover, Czerwiński et al.~\cite{lower-bound} show a quasi-polynomial lower bound for the size of a PG separator.
	Let us also mention that, beside of the four algorithms, there is a fifth QPT algorithm~\cite{ZielonkaParys} obtained by speeding up the Zielonka's recursive algorithm~\cite{Zielonka};
	this algorithm does not fit into the separator approach of Czerwiński et al.~\cite{lower-bound}.

	Of course the idea of converting a parity game into an equivalent safety game is itself much older than QPT algorithms for parity games
	(see e.g.~Bernet, Janin, and Walukiewicz~\cite{parity-to-safety-Walukiewicz}), 
	and was applied not only to finite games, but also to pushdown and collapsible pushdown games~\cite{parity-to-safety-PDA,parity-to-safety-CPDA}.
	
	In this paper we deliberate on the Lehtinen's algorithm~\cite{Lehtinen}.
	As already said, PG separators corresponding to the other algorithms~\cite{calude,small-progress-measure,Fearnley} are deterministic;
	in such a situation it is straightforward that the product game (obtained from an original parity game and the PG separator) is equivalent to the original game 
	(see, e.g.,~\cite[Proposition 3.2]{lower-bound}).
	The PG separator corresponding to the Lehtinen's algorithm~\cite{Lehtinen} is nondeterministic, though, 
	and in general while taking a product of a game with a nondeterministic automaton we do not obtain an equivalent game.
	Actually, a notion of \emph{good-for-games} (GFG) automata was introduced~\cite{good-for-games}; 
	this is a subclass of nondeterministic automata for which it is guaranteed that the product game remains equivalent.
	But one can see that the Lehtinen's separator is not GFG; in consequence, the fact that the Lehtinen's algorithm actually works is quite intriguing.
	As a first contribution we explain this phenomenon.
	Namely, we define a notion of \emph{suitable-for-parity-games} (SFPG) separators, which is more comprehensive that the GFG notion
	(but, unlike GFG, applies only to parity games, not to arbitrary games), and which covers the Lehtinen's separator.
	We then prove that the winner does not change while taking a product of a parity game with an arbitrary SFPG separator,
	which means that every SFPG separator can be used to solve parity games.
	In this way, we establish a framework for solving parity games via nondeterministic PG separators.
	
	As a second contribution, we improve the complexity of the Lehtinen's algorithm.
	Let us recall that the algorithm converts the original parity game with $n$ nodes and $d$ priorities 
	into a parity game with $n^{O(\log d)}$ nodes and $O(\log n)$ priorities (which is actually a product of the original game and of an appropriate SFPG separator).
	Once the new game is created, it has to be solved, say by the small progress measures algorithm~\cite{old-progress-measure}, which is exponential in the number of priorities:
	the resulting complexity is $n^{O(\log d\cdot\log n)}$.\footnote{%
		A better complexity can be obtained by using one of the other QPT algorithms to solve the resulting game.}
	We observe here that the resulting parity game is of a special form---it is possible to win the game without seeing $n$ opponent's priorities in a row---%
	and in consequence it can be solved faster: in time $n^{O(\log n)}$.
	This locates the complexity of the Lehtinen's algorithm much closer to the complexity of the other QPT algorithms~\cite{calude,small-progress-measure,Fearnley,ZielonkaParys},
	which is $n^{O(\log d)}$ (being the same for $d$ close to $n$, but better for games with a small number of priorities).
	
	Our paper is structured as follows.
	In Section~\ref{sec:prelim} we give all necessary definitions.
	In Section~\ref{sec:product} we define SFPG separators, and we prove that they can be used to solve parity games.
	In Section~\ref{sec:lehtinen} we recall the Lehtinen's separator, and we prove that the product game is of a special form.
	In Section~\ref{sec:safety} we prove that this product game can be solved quickly.

\section{Preliminaries}\label{sec:prelim}

\subparagraph{Parity Games.}

	Parity games are played on \emph{game graphs} of the form $G=(V,V_\square,V_\triangle,v_I,E)$, where $V$ is a set of nodes,
	$(V_\square,V_\triangle)$ is a partition of $V$ (which satisfies $V_\square\cup V_\triangle=V$ and $V_\square\cap V_\triangle=\emptyset$),
	$v_I\in V$ is a \emph{starting node}, 
	and $E\subseteq V\times\set{1,2,\dots,d}\times V$ is a set of directed edges labeled by numbers called \emph{priorities}.
	Typically, we assume that $V = \set{1, 2, \dots, n}$ for some natural number $n$. 
	We use $d$ to denote an upper bound for priorities of edges.
	Without loss of generality, we assume that every node has at least one outgoing edge. 
	
	The game is played by two players who are called Even and Odd.   
	A play starts at the starting node $v_I$ and then the players move by following outgoing edges forever, thus forming an infinite path. 
	Every node of the graph is owned by one of the two players: nodes from $V_\square$ and $V_\triangle$ belong to Even and Odd, respectively.
	It is always the owner of the node who moves by following an outgoing edge from the current node to a next one.  
	
	The outcome of the two players interacting in a parity game by making moves is an infinite path in the game graph.  
	We identify such infinite paths with sequences of edges constituting these paths;
	thus an infinite path is an infinite word over the alphabet $\Sigma_{n, d} = \set{1, 2, \dots, n}\times\set{1, 2, \dots, d}\times\set{1, 2, \dots, n}\supseteq E$.
	The set of all infinite words over $\Sigma_{n, d}$ is denoted $\Sigma_{n, d}^\omega$.
	
	We write $\limsupeven_{n, d}$ for the set of infinite words $w\in\Sigma_{n, d}^\omega$ 
	in which the largest number that occurs infinitely many times in the priority component of the letters is even, 
	and we write $\limsupodd_{n, d}$ for the set of infinite words $w\in\Sigma_{n, d}^\omega$ in which that number is odd. 
	Observe that the sets $\limsupeven_{n, d}$ and $\limsupodd_{n, d}$ form a partition of the set $\Sigma_{n, d}^\omega$ of all infinite words over the alphabet~$\Sigma_{n, d}$. 
	An infinite path in a game graph with~$n$ nodes and edge priorities not exceeding~$d$ is \emph{won} by Even if and only if the play is in~$\limsupeven_{n, d}$. 
	
	A \emph{positional strategy} for Even is a set of edges that go out of nodes she owns---exactly one such edge for each of her nodes.
	Even uses such a strategy by always---if the current node is owned by her---following the unique outgoing edge that is in the strategy.
	Note that when Even uses a positional strategy, her moves depend only on the current node---they are oblivious to what choices were made by the players so far. 
	If Even wins the game by following such a strategy, no matter what edges her opponent Odd follows whenever it is her turn to move, then such a strategy is called \emph{winning}.
	Analogously we define a positional (winning) strategy for Odd.
	A basic result for parity games that has notable implications is their \emph{positional determinacy}~\cite{EJ91,Mos91}:
	exactly one of the players has a positional winning strategy. 
	
	The \emph{strategy subgraph} of a game graph $G$ with respect to a positional strategy for Even is the subgraph of $G$ that includes all outgoing edges from nodes owned by Odd 
	and exactly those outgoing edges from nodes owned by Even that are in the positional strategy.
	Observe that the set of plays that arise from Even playing her positional strategy is exactly the set of all plays in the strategy subgraph. 
	
	Let $\poseven_{n,d}$ and $\posodd_{n,d}$ be the sets of all plays that arise from positional winning strategies for Even and Odd, respectively, 
	in some game graph with~$n$ nodes and priorities up to~$d$.    
	Clearly $\poseven_{n, d} \subseteq\limsupeven_{n,d}$ and $\posodd_{n, d} \subseteq\limsupodd_{n,d}$.
	The difference between $\poseven_{n, d}$ and $\limsupeven_{n,d}$ is not only in words that are not valid paths (where the target of some edge does not match the source of the next edge);
	in $\limsupeven_{n,d}\setminus\poseven_{n, d}$ we have for example the path $((1,2,1)(1,1,2)(2,2,2)(2,1,1))^\omega$ (if this path follows a positional strategy for Even in some game graph,
	then $((1,1,2)(2,1,1))^\omega$ follows such a strategy as well, but the latter path is won by Odd).

\subparagraph{Parity and Safety Automata.}

	We consider here only automata reading plays of parity games, so we assume that the input alphabet is $\Sigma_{n,d}$ for some $n$ and $d$.
	We use $d'$ to denote an upper bound for priorities emitted by parity automata.
	A non-deterministic \emph{parity automaton} is a tuple $\Aa=(Q,s_I,\Delta)$,
	where $Q$ is a finite set of \emph{states}, $s_I\in Q$ is an \emph{initial} state,
	and $\Delta\subseteq Q\times\Sigma_{n,d}\times\set{1,2,\dots,d'}\times Q$ is a \emph{transition relation}.
	Without loss of generality, we assume that the transition relation is \emph{total}, that is, 
	for every state~$s$ and letter~$e$, there is some priority $p$ and some state~$s'$, such that the tuple $(s, e, p, s')$ is in the transition relation.
	
	Such a parity automaton can be seen as a directed graph, where $(s, e, p, s')\in\Delta$ is an edge labeled by a letter $e$ and by a priority $p$.
	An infinite path in this graph, starting in the initial state, is called a \emph{run} of $\Aa$.
	The word \emph{read by} such a run (being a word over $\Sigma_{n,d}$) is obtained by projecting every edge of the run to its second component.
	A run is \emph{accepting} if the largest priority that labels infinitely many edges of the run is even.
	If an accepting run reading a word $w$ exists, we say that $w$ is \emph{accepted}, and we write $\Ll(\Aa)$ for the set of all words accepted by $\Aa$.
	
	A parity automaton is called a \emph{safety automaton} if $d'=2$, and there is a set of \emph{rejecting states} such that
	\begin{itemize}
	\item	if $(s,e,1,s')\in\Delta$, then $s'$ is rejecting, and
	\item	if $s$ is rejecting then all transitions $(s,e,p,s')\in\Delta$ are such that $p=1$ and $s'$ is rejecting.
	\end{itemize}
	We notice that a run of a safety automaton is accepting if it does not visit rejecting states.

\section{Product Games and SFPG Separators}\label{sec:product}

	We first recall the notion of product games and separators considered in Czerwiński et al.~\cite{lower-bound}.
	
	\begin{definition}
		Given a game graph $G=(V,V_\square,V_\triangle,v_I,E)$ with at most $n$ nodes and priorities up to~$d$, 
		and a parity automaton~$\Aa=(Q,s_I,\Delta)$ with input alphabet $\Sigma_{n,d}$,
		we define a game graph $G\times\Aa$, called a \emph{synchronized product} of $G$ and $\Aa$, in which
		\begin{itemize}
		\item
			the set of nodes is $(V\cup E)\times Q$, and the starting node is $(v_I,s_I)$;
		\item
			ownership of nodes in $V\times Q$ is inherited from the parity game~$G$, and all nodes in $E\times Q$ belong to Even;
		\item
			for every edge $e=(u,p,v)\in E$ and every state $s\in Q$, there is an edge $((u,s),1,(e,s))$;
		\item
			for every edge $e=(u,p,v)\in E$ and every transition $(s,e,p',s')\in\Delta$, there is an edge $((e,s),p',(v,s'))$;
		\item
			there are no other edges except those specified above.
		\end{itemize}
	\end{definition}
	
	In other words, the players of $G\times\Aa$ play in the parity game~$G$, 
	and the automaton~$\Aa$ is fed the edges corresponding to moves made by the players.
	After every move in $G$, Even resolves non-deterministic choices in $\Aa$.
	In order to win in $G\times\Aa$, Even has to ensure that the run of $\Aa$ reading the play from $G$ is accepting.
	
	It is easy to see that if $\Aa$ is deterministic, and $\Ll(\Aa)$ equals $\limsupeven_{n,d}$ (i.e., the winning condition in $G$), then the games $G$ and $G\times\Aa$ have the same winner.
	The crux of the QPT algorithms is that instead of an automaton recognizing $\limsupeven_{n,d}$, we can use a \emph{PG separator}.
	
	\begin{definition}\label{def:separator}
		Let $\Aa$ be a parity automaton with input alphabet $\Sigma_{n,d}$.
		We say that $\Aa$ is a \emph{parity games separator} (\emph{PG separator}) if it accepts all words from $\poseven_{n,d}$, and rejects all words from $\posodd_{n,d}$.
		If it additionally rejects all words from $\limsupodd_{n,d}$, it is a \emph{strong PG separator}.
	\end{definition}
	
	While for solving parity games (i.e., for the equivalence between $G$ and $G\times\Aa$ described below) it is enough to have a PG separator, the separators corresponding to the
	QPT algorithms~\cite{calude,small-progress-measure,Fearnley,Lehtinen} are actually strong PG separators (cf.~\cite[Section 4]{lower-bound}).
	
	If $\Aa$ is a PG separator, and Odd can win in $G$, then she can also win in $G\times\Aa$: she can ensure that the play from $G$ belongs to $\posodd_{n,d}$,
	and such a play is rejected by $\Aa$.
	The same holds for Even, assuming that $\Aa$ is deterministic.
	If $\Aa$ is nondeterministic, however, it is possible that Even wins in $G$ but Odd wins in $G\times\Aa$.
	Indeed, if Even wins in $G$, she can only ensure that the resulting play is accepted by $\Aa$.
	But in $G\times\Aa$ her task is more difficult: she has to resolve nondeterministic choices of $\Aa$ as they arise, without knowing the whole play from $G$.
	The abilities of Even are described by transition strategies.
	
	\begin{definition}
		A \emph{transition strategy} for an automaton $\Aa$ is a function $\sigma\colon\Sigma_{n,d}^*\times Q\times\Sigma_{n,d}\to\Delta$
		such that $\sigma(w,s,e)$ is of the form $(s,e,p,s')$ for all $(w,s,e)\in\Sigma_{n,d}^*\times Q\times\Sigma_{n,d}$.
		We use such a strategy to resolve non-deterministic choices: if the word read so far is $w$, the state of $\Aa$ is $s$, and the next letter to be read is $e$, 
		then we proceed using the transition $f(w,s,e)$.
		We say that a transition strategy $\sigma$ is \emph{winning} for a set of words $L\subseteq \Ll(\Aa)$ if
		for every word $w\in L$, the run obtained by following $\sigma$ while reading the word $w$ is accepting.
	\end{definition}
	
	Henzinger and Piterman~\cite{good-for-games} proposed a notion of good-for-games automata:
	an automaton $\Aa$ is \emph{good for games} (\emph{GFG}) if in $\Aa$ there exists a transition strategy that is winning for $\Ll(\Aa)$.
	If $\Aa$ is GFG, then Even can use a winning strategy from $G$ and a transition strategy winning for $\Ll(\Aa)$ to win in $G\times\Aa$.
	We observe, though, that it is not a problem for Even to have a transition strategy that depends on $G$, and on her winning strategy in $G$.
	This way we come to a more comprehensive definition of SFPG separators.

	\begin{definition}\label{def:suitable}
		A PG separator $\Aa$ with input alphabet $\Sigma_{n,d}$ is \emph{suitable for parity games} (\emph{SFPG}) if
		for every game graph $G$ with $n$ nodes and priorities up to $d$, and for every positional winning strategy $\tau$ for Even in $G$,
		the automaton $\Aa$ has a transition strategy $\sigma$ winning for the set of all plays in $G$ that arise from $\tau$.
	\end{definition}
	
	Notice that every deterministic automaton $\Aa$ is good for games: the transition strategy that in every situation chooses the only available transition allows to accept all words from $\Ll(\Aa)$.
	Moreover, every good-for-games PG separator $\Aa$ is SFPG: for every $G$ and $\tau$ as in Definition~\ref{def:suitable}, 
	all plays in $G$ that arise from $\tau$ are accepted by $\Aa$ (because $\Aa$ is a PG separator), 
	and thus the transition strategy that is winning for the whole $\Ll(\Aa)$ (existing because $\Aa$ is GFG) can be used for the set of these plays.
	In the next section we present the PG separator corresponding to Lehtinen's algorithm; it is neither deterministic nor good for games, but it is SFPG.
	
	We now prove that by producting a parity game with an SFPG separator, we obtain an equivalent game.
	
	\begin{theorem}\label{thm:safety-game-equiv}
		If $G$ is a game graph with $n$ nodes and priorities up to $d$,
		and $\Aa$ is an SFPG separator with input alphabet $\Sigma_{n,d}$,
		then Even has a winning strategy in $G$ if and only if she has a winning strategy in the synchronized product $G\times\Aa$.
	\end{theorem}
	
	\begin{proof}
		Suppose first that Even has a winning strategy in $G$.
		Then, by positional determinacy, she also has a positional winning strategy $\tau$ in $G$.
		Because the separator $\Aa$ is SFPG, it has a transition strategy $\sigma$ that is winning for the set of all plays in $G$ that arise from $\tau$.
		Using $\tau$ and $\sigma$ we define an Even's strategy in $G\times\Aa$:
		she plays according to $\tau$ in the $G$ component, and according to $\sigma$ in the $\Aa$ component.
		An infinite play of $G\times\Aa$ following this strategy is a pair: a play $w$ in $G$ following $\tau$, and a run $\rho$ of $\Aa$ reading $w$ and following $\sigma$.
		By assumption on $\sigma$, because $w$ is a play in $G$ that arises from $\tau$, we obtain that $\rho$ is accepting.
		This implies that the considered play of $G\times\Aa$ is won by Even, and thus the considered strategy is winning for Even.
		
		Next, suppose that Even does not have a winning strategy in $G$.
		Then, by positional determinacy, Odd has a positional winning strategy $\tau$ in $G$.
		This strategy can be also used in $G\times\Aa$, as Odd takes decisions only in the $G$ part of $G\times\Aa$.
		Consider a play of $G\times\Aa$ following this strategy; it consists of a play $w$ in $G$ following $\tau$, and of a run $\rho$ of $\Aa$ reading $w$.
		Because $\tau$ is a positional winning strategy for Odd, we have $w\in\posodd_{n,d}$, hence, because the PG separator $\Aa$ rejects all words from $\posodd_{n,d}$,
		the run $\rho$ is rejecting; the play is won by Odd.
		This implies that Even does not have a winning strategy in $G\times\Aa$.
	\end{proof}
	
	Notice that the product game $G\times\Aa$ is larger, but potentially simpler, than $G$.
	For example, if $\Aa$ is a safety automaton, out of a parity game we obtain a safety game; the latter can be solved in linear time.

	We remark that Colcombet and Fijalkow in their recent work~\cite{universal-graphs} define a similar notion of good-for-small-games automata:
	an automaton $\Aa$ is good for $(n,d)$-parity games if it satisfies our Theorem~\ref{thm:safety-game-equiv}, that is,
	if for every game graph $G$ with $n$ nodes and priorities up to $d$, the games $G$ and $G\times\Aa$ are equivalent.
	Such a definition is purely semantical; it does not give any hint which automata are indeed good for $(n,d)$-parity games.
	Our definition of SFPG separators is more concrete: it specifies particular conditions on an automaton (what it should accept / reject, and in which way).
	In Theorem~\ref{thm:safety-game-equiv} we then prove that every SFPG separator can be indeed used to solve parity games 
	(i.e., that it is good for $(n,d)$-parity games, in the terminology of Colcombet and Fijalkow).

\section{Register Automata}\label{sec:lehtinen}

	In this section we express Lehtinen's construction~\cite{Lehtinen} as an SFPG separator $\Rr_{n,d}$.
	Recall that out of a parity game $G$ she constructs an equivalent parity game, which is essentially $G\times\Rr_{n,d}$.
	
	The idea of the construction is to store some recently visited priorities in some number of registers.
	Let us denote $\rn(n)=1+\floor{\log_2 n}$; this is the number of registers needed to solve games with $n$ nodes.
	In Lehtinen's work, $\rn(n)$ is called a register index.\footnote{%
		While there exist games with $n$ nodes that can be solved using less registers (i.e., games with a smaller register index),
		$1+\floor{\log_2 n}$ is the upper bound.}
	
	For all positive numbers~$n$ and~$d$, such that~$d$ is even, we define
	a \emph{non-deterministic parity automaton}~$\Rr_{n, d}$ in the
	following way. 
	\begin{itemize}
	\item
		The set of states of~$\Rr_{n, d}$ is the set of 
		non-increasing $\rn(n)$-sequences 
		$\seq{r_{\rn(n)}, \dots, r_2, r_1}$ of ``registers'' that hold
		numbers in $\set{1, 2, \dots, d}$.  
		The initial state is $\seq{1, 1, \dots, 1}$. 
	
	\item
		For every state $s=\seq{r_{\rn(n)}, \dots, r_2, r_1}$
		and
		letter $e = (u, p, v) \in \Sigma_{n, d}$, we define the 
		\emph{update of~$s$ by~$e$} to be the state
		$\seq{r_{\rn(n)}, \dots, r_{k+1}, p, \dots, p}$, where $k$ is the
		greatest index such that $r_1,\dots,r_k < p$.   
	
	\item
		For every state $s=\seq{r_{\rn(n)}, \dots, r_2, r_1}$
		and for every~$k$, $1 \leq k \leq \rn(n)$, we define the  
		\emph{$k$-reset of~$s$} to be the state 
		$\seq{r_{\rn(n)}, \dots, r_{k+1}, r_{k-1}, \dots, r_2, 1}$. 
		We say that this $k$-reset is even (odd) if $r_k$ is even (odd, respectively).
	
	\item
		For every state~$s$ and letter $e\in \Sigma_{n, d}$, if $s'$ is the update of~$s$ by~$e$,
		then in the transition relation there is a transition $(s, e, 1, s')$, called a \emph{non-reset} transition. 
	
	\item
		For every state~$s$, 
		letter $e \in \Sigma_{n, d}$, 
		and for every~$k$, $1 \leq k \leq \rn(n)$, 
		if $s'$ is the update of~$s$ by~$e$,
		and $s''$ is the even $k$-reset of~$s'$,
		then in the transition relation there is a transition $(s, e, 2k, s'')$, called an \emph{even reset of register~$k$}. 
	
	\item
		For every state~$s$,
		letter $e \in \Sigma_{n, d}$, 
		and for every~$k$, $1 \leq k \leq \rn(n)$, 
		if $s'$ is the update of~$s$ by~$e$,
		and $s''$ is the odd $k$-reset of~$s'$,
		then in the transition relation there is a transition $(s, e, 2k+1, s'')$, called an \emph{odd reset of register~$k$}. 
	
	\item
		There are no other transitions in $\Rr_{n,d}$ except those specified above.
	\end{itemize}
	
	In Theorem~\ref{thm:ok} we prove that $\Rr_{n,d}$ is indeed an SFPG separator.
	Moreover, we prove that its runs are of a special form, as specified by Definition~\ref{def:bad};
	this is useful in Section~\ref{sec:safety}, where we argue that the product game $G\times\Rr_{n,d}$ can be solved faster than an arbitrary parity game.
	
	\begin{definition}~\label{def:bad}
		Let $\rho$ be a run of a parity automaton.
		We define $\bad(\rho)$ to be the greatest number $m$ such that in $\rho$ there is an infix containing $m$ transitions emitting some odd priority $p$ and no transitions emitting higher priority.
	\end{definition}
	
	\begin{theorem}\label{thm:ok}
		The automaton $\Rr_{n, d}$ is a strong SFPG separator.
		Moreover, for every game graph $G$ with~$n$ nodes and priorities up to~$d$, for every Even's positional winning strategy $\tau$ in $G$,
		and for every run $\rho$ of $\Rr_{n, d}$ that follows the transition strategy existing by Definition~\ref{def:suitable} and that reads a play in $G$ arising from $\tau$,
		it holds that $\bad(\rho)\leq n-1$.
	\end{theorem}
	
	We now prove Theorem~\ref{thm:ok}.
	We start with the easier part, saying that $\Aa$ rejects all words from $\limsupodd_{n,d}$.
	Consider thus a word $w\in\limsupodd_{n,d}$, and a run $\rho$ of $\Rr_{n,d}$ reading this word.
	If from some moment there are no more resets in this run, then indeed $\rho$ is rejecting.
	Otherwise, consider the greatest (odd) priority $p$ occurring in $w$ infinitely often, 
	and consider the greatest index $k$ such that there are infinitely many resets of register $k$ in $\rho$.
	From some moment on, in $\rho$ no priority higher than $p$ is read, and there is no reset of any register $l>k$.
	A little bit later, after $k$ resets of register $k$, the value of register $k$ is at most $p$ for the rest of the run.
	Then, infinitely many times the priority $p$ is read, it is stored to register $k$, never overwritten by anything larger, and then reset.
	This means that there are infinitely many odd resets of register $k$, emitting priority $2k+1$, while no higher priorities are emitted (except in the finite prefix that we have skipped).
	In consequence, $\rho$ is rejecting.
		
	For the remaining part of the proof, fix a game graph $G$, and an Even's positional winning strategy $\tau$.
	Let $G_\tau$ be the strategy subgraph of $G$ with respect to $\tau$, and let $V_\tau$ be the set of those nodes of $G_\tau$ that are reachable from the starting node.
	For a priority $p$, and for $S\subseteq V_\tau$, let $G_{S,p}$ be the subgraph of $G_\tau$ that contains only nodes that belong to $S$ and only edges of priority not larger than $p$.
	
	We now define a \emph{game tree} of $G_\tau$ in a top-down fashion.
	The root of this tree is $(\ceil{d/2},V_\tau)$.
	Let now $(k,S)$ be an (already defined) node of this tree such that $k\geq 1$,
	and let $S_1,S_2,\dots,S_m$ be (the sets of nodes of) all the strongly connected components of $G_{S,2k-1}$.
	We assume that $S_1,S_2,\dots,S_m$ are sorted topologically, that is, that in $G_{S,2k-1}$ there are no edges to $S_i$ from $S_j$ when $i<j$ 
	(if there are multiple such orders of $S_1,S_2,\dots,S_m$, we fix one of them).
	In such a case, $(k-1,S_1),(k-1,S_2),\dots,(k-1,S_m)$ are children of $(k,S)$, in this order.
	\begin{figure*}
		\begin{center}
			\import{pics/}{game-tree.pdf_tex_ok}
		\end{center}
		\caption{An example of a strategy subgraph, together with the corresponding game tree.
		Dashed circles depict nodes of the game tree: the largest circle is the root $(3,S)$, inside it we have two children $(2,S_1), (2, S_2)$, ordered left to right, and so on; 
		additionally, every node $x$ of the graph also constitutes a node $(0,\{x\})$ (i.e., a leaf) of the game tree.
		Notice that edges with odd priorities can only go right.
		This is a strategy subgraph, so nodes belonging originally to Even have here only a single successor.
		}
		\label{fig:game-tree}
	\end{figure*}
	An example of a game tree is presented in Figure~\ref{fig:game-tree}.
	
	Notice that if $S_i$ is a strongly connected component of $G_{S,2k-1}$, then it does not contain edges of priority $2k-1$.
	Indeed, if such an edge existed inside a strongly connected component,
	there would be a cycle in $G_{S,2k-1}$ (i.e., in $G_\tau$) on which the maximal priority would be $2k-1$ (odd);
	by reaching such a cycle (recall that $S\subseteq V_\tau$ contains only nodes reachable in $G_\tau$ from the starting node) and repeating it forever, 
	we would obtain a play won by Odd, while all plays in $G_\tau$ are, by assumption, won by Even.
	It follows that $S_i$ is actually also a strongly connected component of $G_{S,2k-2}$.
	In other words, $G_{S_i,2k-2}=G_{S_i,2k-1}$.
	
	For a node $(k,S)$ of the game tree, and for $l<k$, let $\fst_l(S)=S'$ for $(l,S')$ being the leftmost descendant of $(k,S)$ located on level $l$.
	
	The following lemma states our thesis in a form suitable for induction.
	It uses a notion of a \emph{partial run}, which is defined like a run, but it needs not to start in the initial state, and it needs not to be infinite.
	
	\begin{lemma}\label{lem:ind-ass}
		Let $(k,S)$ be a node of the game tree of $G_\tau$, 
		let $s$ be a state of $\Rr_{n,d}$,
		and let $\xi$ be a nonempty (finite or infinite) path in $G_{S,2k}$ starting in a node $v$.
		Assume that if an odd number $2l+1$ (where $l\geq 1$) is contained in some of the registers $1,2,\dots,\rn(|S|)$ of $s$, 
		then $l<k$ and $v\not\in\fst_l(S)$.
		Under these assumptions, there exists a partial run $\rho$ from $s$ reading $\xi$, such that 
		\begin{enumerate}
		\item	in $\rho$ there are no resets of registers above $\rn(|S|)$,
		\item	if the register $\rn(|S|)$ in $s$ contains an even number not smaller than $2k$, then in $\rho$ there are no odd resets of the register $\rn(|S|)$,
		\item	$\bad(\rho)\leq|S|-1$, and
		\item	$\rho$ follows a transition strategy (that may depend on $G,\tau,k,S,s$): in every step, 
			the non-determinism is resolved basing only on the prefix of $\xi$ read so far and on the next edge of $\xi$ that should be read.
		\end{enumerate}
	\end{lemma}
	
	In order to finish the proof of Theorem~\ref{thm:ok}, we simply use Lemma~\ref{lem:ind-ass} for $(k,S)=(\ceil{d/2},V_\tau)$, and for $s=\seq{1, 1, \dots, 1}$ 
	(i.e., for the initial state of $\Rr_{n,d}$).
	Indeed, every play $w$ in $G$ that arises from the strategy $\tau$ is a path in $G_{V_\tau,2\ceil{d/2}}$;
	thus the lemma gives is a run $\rho$ reading $w$.
	By Point 3, $\bad(\rho)$ is finite, which implies that $\rho$ is accepting.
	Because this holds for all $G$ and $\tau$, and because $\Aa$ rejects all words from $\limsupodd_{n,d}$ (as shown at the beginning), we already know that $\Aa$ is a strong PG separator.
	Point 4 says that $\rho$ is constructed following a transition strategy, so $\Aa$ is SFPG.
	The condition $\bad(\rho)\leq|V_\tau|-1$ from Point 3 gives us the second part of the theorem's statement.
	
	\begin{proof}[Proof of Lemma~\ref{lem:ind-ass}]
		We proceed by induction on $k$.
		If $k=0$, then there is no nonempty path $\xi$ in $G_{S,2k}$ (this graph has no edges), so the lemma trivially holds.
	
		For the rest of the proof, suppose that $k\geq 1$.
		Let $(k-1,S_1),\dots,(k-1,S_m)$ be the children of $(k,S)$.
		By definition, $S_1,\dots,S_m$ form a division of $S$.
		Obviously $|S_i|\leq|S|$, so $\rn(|S_i|)\leq\rn(|S|)$, for all $i\in\{1,\dots,m\}$.
		
		Notice first that no matter how $\rho$ is constructed, none of its last $\rn(|S|)$ registers contains an odd number greater than $2k$, in all states of $\rho$%
		---call this property ($\spadesuit$).
		This holds because the condition is satisfied in the first state $s$ of $\rho$, and then only edges of priority up to $2k$ are read.
		
		We construct a run $\rho$ reading $\xi$ by repeating the following steps:
		\begin{itemize}
		\item	in the remaining part of $\xi$, let $\xi'$ be the maximal prefix that stays in $G_{S_i,2k-1}$ (i.e., in $G_{S_i,2k-2}$) for some $i$ 
			(possibly $|\xi'|=0$, i.e., already the first edge leaves $G_{S_i,2k-1}$);
		\item	if $i\geq 2$, then
			\begin{itemize}
			\item	let $j_1<j_2<\dots<j_r$ be the numbers of registers among $1,\dots,\rn(|S_i|)$ which, in the current state, contain an odd priority higher than $1$;
			\item	while reading the first $\min(r,|\xi'|)$ edges of $\xi'$, we perform resets of the registers $j_1, j_2, \dots,\allowbreak j_{\min(r,|\xi'|)}$, consecutively%
			---call these transitions \emph{preparatory transitions};
			\end{itemize}
		\item	let $\xi''$ be the part of $\xi'$ that remains to be read;
		\item	if $|\xi''|>0$, then we use the induction assumption with $k-1$ as $k$ and with $S_i$ as $S$ to construct a fragment of a run that reads $\xi''$
			(we prove below that the induction assumption can indeed be used)---call the fragment of $\rho$ obtained this way a \emph{block of local transitions};
		\item	we have now read the whole $\xi'$;
		\item	if $\xi$ already ended, we stop the construction;
		\item	otherwise, the next edge of $\xi$ leads outside $G_{S_i,2k-1}$;
		\item	if this edge has priority $2k$, we reset the register $\rn(|S|)$ while reading this edge---call this a \emph{valuable transition};
		\item	otherwise, we perform a non-reset transition reading this edge---call this a \emph{regressive transition};
		\item	we repeat the procedure from the beginning.
		\end{itemize}
		
		We have to prove that indeed the induction assumption can be used above.
		To this end, consider the state $s'$ from which we are about to start a block of local transitions reading a path $\xi''$ in $G_{S_i,2k-2}$,
		and let $v'$ be the first node of this path.
		Suppose that some of the registers $1,2,\dots,\rn(|S_i|)$ of $s'$ contains an odd number $2l+1$, where $l\geq 1$.
		We have to prove that $l<k-1$, and that $v'\not\in\fst_l(S_i)$.
		By Property ($\spadesuit$), $l<k$.
		There are three cases:
		\begin{itemize}
		\item	Suppose that $i=1$ and this is the first time when the loop is used.
			Then there are no preparatory transitions, so $s'=s$ and $v'=v$.
			By assumptions of the lemma, $v\not\in\fst_l(S)$.
			If $l=k-1$, we would have $v\not\in\fst_l(S)=S_1$, while $v\in S_i=S_1$; thus $l<k-1$.
			We then have $v'=v\not\in\fst_l(S)=\fst_l(S_1)$.
		\item	Suppose that $i=1$ and $\xi''$ is preceded in $\xi$ by some edge.
			This edge is not an edge of $G_{S_1,2k-1}$, by maximality of the previous block of local transitions.
			By the definition of a game tree, there are no edges in $G_{S,2k-1}$ coming to $S_1$ from $S\setminus S_1$
			($S_1,\dots,S_m$ are topologically sorted strongly connected components of $G_{S,2k-1}$).
			Thus, the edge preceding $\xi''$ has priority $2k$.
			After reading this edge, all registers contain value $2k$ or higher, or $1$ (if there was a reset);
			they cannot contain $2l+1$ with $1\leq l<k$.
		\item	Otherwise, $i\geq 2$.
			Then, we are just after preparatory transitions.
			All odd values (greater than $1$) present before these transitions were reset to $1$.
			Thus, priority $2l+1$ appears in a register of $s'$ because it was read during preparatory transitions,
			and later no edges with priority higher than $2l+1$ were read.
			This already implies that $l<k-1$, because only edges of priority up to $2k-2$ are read during preparatory transitions.
			Edges of priority up to $2l+1$ cannot lead to $\fst_l(S_i)$ from $S_i\setminus\fst_l(S_i)$:
			by the definition of the game tree, for every level $j$ with $l\leq j\leq k-2$, there are no edges of priority up to $2j+1$ leading to $\fst_j(S_i)$ from its (following) siblings.
			Moreover, there are no edges of priority $2l+1$ inside $\fst_l(S_i)$.
			Thus, after reading an edge of priority $2l+1$, and then some edges of priority up to $2l+1$, we cannot end inside $\fst_l(S_i)$.
		\end{itemize}
	
		We now have to check Points 1-4 from the statement of the lemma.
		Point 1 is immediate: preparatory and valuable transitions reset only registers up to $\rn(|S|)$, regressive transitions do not reset anything,
		and local transitions, by Point 1 of the induction assumption, also reset only registers up to $\rn(|S|)$ (recall that $\rn(|S_i|)\leq\rn(|S|)$).
		
		Point 4 is also immediate: by definition we create $\rho$ in a deterministic way.
	
		While proving Points 2-3 we assume that $|S|\geq 2$;
		the degenerate case of $|S|=1$ is handled at the very end.
		
		We now prove Point 2 saying that in $\rho$ there are no odd resets of the register $\rn(|S|)$ if this register in the first state of $\rho$ contains an even number not smaller than $2k$.
		Simultaneously, we prove that odd resets of the register $\rn(|S|)$ can appear in $\rho$ only before the first valuable transition---call this property ($\clubsuit$).
		Notice first that when we visit some $S_i$ such that $\rn(|S_i|)<\rn(|S|)$, 
		then neither preparatory transitions, nor local transitions (by Point 1 of the induction assumption) reset the register $\rn(|S|)$.
		On the other hand, $\rn(|S_i|)=\rn(|S|)$ implies that $|S_i|>|S|/2$, which is possible only for one component $S_i$; call it $S_{\max}$.
		Regressive transitions do not reset anything.
	
		It remains to handle valuable transitions, and transitions reading edges from $G_{S_{\max},2k-2}$ in the case of $\rn(|S_{\max}|)=\rn(|S|)$;
		these transitions may reset the register $\rn(|S|)$.
		Recall that, in all states of $\rho$, none of the last $\rn(|S|)$ registers can contain an odd number greater than $2k$ (Property ($\spadesuit$)).
		Consider a valuable transition.
		After the update by priority $2k$, the registers $\rn(|S|)$ and $\rn(|S|)-1$ contain even numbers not smaller than $2k$ 
		(we put there $2k$ during the update, unless a larger even priority is already there).
		Thus, when we reset the register $\rn(|S|)$ during a valuable transition, its value is even.
		Moreover, after this transition, the register $\rn(|S|)$ still contains an even number not smaller than $2k$,
		moved there from the register $\rn(|S|)-1$ (here it is important that $\rn(|S|)\geq 2$, so that the register $\rn(|S|)-1$ indeed exists).
	
		By the definition of a game tree, if we leave $G_{S_{\max},2k-1}$, then before entering $G_{S_{\max},2k-1}$ again there is 
		an edge of priority $2k$, resulting in a valuable transition (edges of priority up to $2k-1$ cannot go to $S_i$ from $S_j$ when $i<j$).
		Assuming that the register $\rn(|S|)$ of $s$ (i.e., of the state from which we start $\rho$) contains an even number not smaller than $2k$,
		it follows that whenever we reach $S_{\max}$, the register $\rn(|S|)$ contains an even number not smaller than $2k$
		(either existing there from the beginning of $\rho$, or since the last valuable transition).
		Thus, the register $\rn(|S|)$ is not reset during preparatory transitions, and by Point 2 of the induction assumption, 
		there are no odd resets during the considered block of local transitions for $S_{\max}$; we obtain Point 2.
		
		For Property ($\clubsuit$), we do not have the assumption that at the very beginning the register $\rn(|S|)$ contains an even number not smaller than $2k$.
		In consequence, there may be odd resets of the register $\rn(|S|)$ while $S_{\max}$ is visited for the first time, but later, after the first valuable transition, such resets are again impossible.
	
		Next, concentrate on Point 3.
		We need to prove that:
		\begin{itemize}
		\item	for every $r$, in every infix of $\rho$ without resets of registers above $r$, there are at most $|S|-1$ odd resets of the register $r$, and
		\item	in every infix of $\rho$ without any resets, there are at most $|S|-1$ (non-reset) transitions.
		\end{itemize}
		For $r>\rn(|S|)$ there are no $r$-resets at all (Point 1).
		Take some $r\leq\rn(|S|)$, and consider an infix $\rho'$ of $\rho$ without any resets of registers above $r$;
		let $\xi'$ be the path read by $\rho'$.
		We are about to bound the number of odd resets of the register $r$ in $\rho'$.
		If $r<\rn(|S|)$, the infix $\rho'$ does not contain valuable transitions, as they reset the register $\rn(|S|)$, being above the register $r$.
		If $r=\rn(|S|)$, we can also assume that $\rho'$ does not contain valuable transitions, as anyway, by Property ($\clubsuit$), 
		after the first valuable transition there are no more odd resets of the register $\rn(|S|)$.
		In consequence, $\xi'$ is a path in $G_{S,2k-1}$.
		By the definition of a game tree, such a path can visit components $S_1,S_2,\dots,S_m$ only in an ascending order;
		every $G_{S_i,2k-1}$ is visited by $\xi'$ at most once.
		By Point 3 of the induction assumption, in the block of local transitions in $\rho'$ visiting $S_i$ 
		there are at most $|S_i|-1$ odd resets of the register $r$ without any resets of registers above $r$ in between.
		Moreover, in every block of preparatory transitions, we reset the register $r$ at most once, and there are $m-1$ such blocks: before $S_2,S_3,\dots,S_m$, but not before $S_1$.
		Together, there are at most $\sum_{i=1}^m(|S_i|-1)+m-1=|S|-1$ odd resets of the register $r$ in $\rho'$, as wanted.
	
		The situation is similar when we consider an infix $\rho'$ of $\rho$ without any resets, and we want to bound its length.
		Again, it visits every $G_{S_i,2k-1}$ at most once.
		In every $S_i$ there are at most $|S_i|-1$ non-reset transitions in a row, by Point 3 of the induction assumption, and we have at most $m-1$ regressive transitions.
		
		This finishes the proof when $|S|\geq 2$.
		It remains to prove Points 2-3 in the degenerate case of $|S|=1$, when $\rn(|S|)=1$.
		In this case, all edges in $G_{S,2k}$ are loops around the only node in $S$.
		None of them can have an odd priority, because by reaching this node and then repeating this loop we would obtain a play won by Odd, while by assumption all plays in $G_\tau$ are won by Even.
		Moreover, by assumption, if the register $1$ of $s$ (i.e., of the state from which we start $\rho$) contained an odd number $2l+1$,
		then $l<k$ and $v\not\in\fst_l(S)$.
		But, because $|S|=1$, we have $\fst_l(S)=S$, and $v\in S$.
		Thus, the register $1$ of $s$ contains an even number.
		In consequence, in all states of $\rho$ the last register (the register number $\rn(|S|)$) contains either an even number or $1$;
		we then update it by an even number (so it cannot contain $1$ after this update), and then we possibly reset it.
		This means that we can only have even resets of the register $\rn(|S|)$, which gives Point 2.
		For Point 3, we also need to know that there are no non-reset transitions.
		But observe that for $|S|=1$ there are no regressive transitions (all edges of priority up to $2k-1$ stay inside $G_{S_1,2k-1}$),
		and there are no non-reset transitions among local transitions, by Point 3 of the induction assumption.
	\end{proof}
	
	We remark that the proof presented above is based on Lehtinen's work~\cite{Lehtinen}.
	We only prove a slightly stronger property, and we expand some details that in Lehtinen's paper are treated in a quite sketchy way.
	
	Because states of $\Rr_{n,d}$ consist of non-increasing $\rn(n)$-sequences of priorities in $\{1,\dots,d\}$, and because $\rn(n)=1+\floor{\log_2 n}$,
	the number of states of $\Rr_{n,d}$ is
	\begin{align*}
		\eta_{n,d} = \binom{\rn(n)+d-1}{\rn(n)}=d^{O(\log n)}=n^{O(\log d)}\,;
	\end{align*}
	from every state the automaton has $\rn(n)+1$ transitions reading every letter $e\in\Sigma_{n,d}$ (a non-reset transition, and a reset transition for every register).
	In consequence, for a game graph $G$ with $n$ nodes, $m$ edges, and priorities up to $d$, the product game $G\times\Rr_{n,d}$ has $(n+m)\cdot\eta_{n,d}$ nodes,
	$m\cdot\eta_{n,d}\cdot(\rn(n)+2)$ edges, and uses $2\cdot \rn(n)+1$ priorities.
	Using a standard (i.e., not quasi-polynomial-time) algorithm to solve such a game, the number of priorities goes to the exponent, thus we obtain complexity $n^{O(\log d\cdot\log n)}$.

\section{Safety Register Automata}\label{sec:safety}
	
	In the final section we show that the property $\bad(\rho)\leq n-1$ obtained in Theorem~\ref{thm:ok} allows us to solve the product game $G\times\Rr_{n,d}$ faster: 
	in time $n^{O(\log n)}$ instead of $n^{O(\log d\cdot\log n)}$.
	We could prove this directly, but instead we modify the parity automaton $\Rr_{n,d}$ into a \emph{safety} automaton $\Ss_{n,d}$.
	
	We define the safety automaton~$\Ss_{n, d}$ in the following way:
	\begin{itemize}
	\item
		The set of states of~$\Ss_{n, d}$ is the set of pairs:
		the first component is a state of the automaton~$\Rr_{n, d}$ and the
		other component is an $(\rn(n)+1)$-sequence 
		$\seq{c_{\rn(n)}, \dots, c_1, c_0}$ of \emph{counters} with
		values in $\set{1, \dots, n}$;
		additionally, in $\Ss_{n, d}$ we have a designated rejecting state $\rej$.  
	
		Throughout this definition, $c$ always stands for the
		sequence $\seq{c_{\rn(n)}, \dots, c_1, c_0}$. 
	
	\item
		The initial state is $(s^0, c^0)$, where $s^0$ is the initial state
		of~$\Rr_{n, d}$ and $c^0 = \seq{n, n, \dots, n}$. 
	
	\item
		For each transition $(s, e, 2k, s')$ in~$\Rr_{n, d}$ that is an even
		reset of the register~$k$, we have a transition 
		$\big((s, c), e, 2, (s', c')\big)$ in~$\Ss_{n, d}$, 
		where 
		$c' = \seq{c_{\rn(n)}, \dots, c_{k+1}, c_k, n, \dots, n}$. 
		
	\item
		For each transition $(s, e, 2k+1, s')$ in~$\Rr_{n, d}$ that has an odd priority (i.e., is a non-reset transition, or is an odd reset of the register~$k$), 
		we have a transition 
		$\big((s, c), e, 2, (s', c')\big)$ in~$\Ss_{n, d}$, 
		where 
		$c' = \seq{c_{\rn(n)}, \dots, c_{k+1}, c_k - 1, n, \dots, n}$, 
		if $c_k > 1$. 
	
	\item
		For each transition $(s, e, 2k+1, s')$ in~$\Rr_{n, d}$ that has an odd priority
		we have a transition 
		$\big((s, c), e, 1, \rej\big)$ in~$\Ss_{n, d}$, 
		where $c_k=1$.
	
	\item
		Moreover, for every letter $e$, we have a transition $(\rej, e, 1, \rej)$ in~$\Ss_{n, d}$.

	\item
		There are no other transitions in~$\Ss_{n, d}$ except those specified above.
	\end{itemize}
	
	\begin{theorem}\label{thm:s-ok}
		The automaton $\Ss_{n, d}$ is a strong SFPG separator.
	\end{theorem}
	
	\begin{proof}
		Consider first a word $w\in\limsupodd_{n,d}$, and a run $\rho_\Ss$ of $\Ss_{n,d}$ reading this word;
		we have to prove that $\rho_\Ss$ is rejecting.
		While projecting every state of $\rho_\Ss$ to its first component, we obtain a run $\rho_\Rr$ of $\Rr_{n,d}$ also reading $w$.
		By Theorem~\ref{thm:ok}, $\Rr_{n,d}$ is a strong SFPG separator, so it rejects all words from $\limsupodd_{n,d}$ (cf.~Definition~\ref{def:separator}); $\rho_\Rr$ is rejecting.
		Let $p$ be the largest priority emitted by $\rho_\Rr$ infinitely often; $p$ is odd.
		Consider the suffix of $\rho_\Rr$ in which no larger priority is emitted.
		Concentrate now on $\rho_\Ss$.
		Emitting the priority $p$ by $\rho_\Rr$ results in decreasing the counter $(p-1)/2$ by $1$,
		while emitting priorities lower than $p$ leaves the counter $(p-1)/2$ unchanged.
		Thus, after $n$ transitions emitting the priority $p$ the rejecting state is reached; $\rho_\Ss$ is rejecting.
		
		Next, consider a game graph $G$ with~$n$ nodes and priorities up to~$d$, and an Even's positional winning strategy $\tau$ in $G$.
		Let $\sigma_\Rr$ be the transition strategy in $\Rr_{n,d}$ that is winning for the set of all plays in $G$ that arise from $\tau$, 
		existing because $\Rr_{n,d}$ is SFPG (cf.~Definition~\ref{def:suitable}).
		We extend $\sigma_\Rr$ to a transition strategy $\sigma_\Ss$ for automaton $\Ss_{n,d}$: 
		in $\sigma_\Ss$ we resolve nondeterministic choices in the same way as in $\sigma_\Rr$; 
		the difference is only that in $\Ss_{n,d}$ we additionally update the counters (in a deterministic way).
		We have to prove that $\sigma_\Ss$ is winning for the set of all plays in $G$ that arise from $\tau$.
		To this end, consider such a play; let $\rho_\Ss$ be the run of $\Ss_{n,d}$ that follows $\sigma_\Ss$ and reads this play.
		Let $\rho_\Rr$ be the corresponding run of $\Rr_{n,d}$, obtained by projecting every state of $\rho_\Ss$ to its first component.
		By Theorem~\ref{thm:ok}, $\bad(\rho_\Rr)\leq n-1$.
		Notice that when $\rho_\Rr$ emits an odd priority $2k+1$, we decrease the counter $k$ by $1$,
		and when it emits any higher priority, we reset the counter $k$ to $n$.
		Because priority $2k+1$ is emitted at most $n-1$ times without emitting any higher priority in between (by the condition $\bad(\rho_\Rr)\leq n-1$),
		we obtain that the rejecting state is not reached.
		In consequence $\rho_\Ss$ is accepting, which finishes the proof.
	\end{proof}
	
	We see that $\Ss_{n,d}$ has $\xi_{n,d}=\eta_{n,d}\cdot n^{\rn(n)+1}+1$ states, and that from every state it has at most $\rn(n)+1$ transitions reading every letter.
	Thus, for a game graph $G$ with $n$ nodes, $m$ edges, and priorities up to $d$, the product game $G\times\Ss_{n,d}$ has $(n+m)\cdot\xi_{n,d}$ nodes 
	and no more than $m\cdot\xi_{n,d}\cdot(\rn(n)+2)$ edges.
	This safety game can be solved in linear time.
	Without loss of generality we can assume that $d\leq n$, so the running time is of the form $n^{O(\log n)}$.
	
	\begin{remark}\label{rem:1}
		Let us underline two aspects of the definition of SFPG separators that are important for the proofs of Theorems~\ref{thm:ok} and~\ref{thm:s-ok}.
		First, the transition strategies that we create actually depend on $G$ and $\tau$ (unlike in the definition of good-for-games automata).
		Second, in our transition strategies we choose a next transition basing not only on the priority of an edge to be read, but also basing on its target.
		For this reason, we use automata that read edges (i.e., triples: source, priority, target), not just priorities, 
		nor pairs: source node of an edge, priority of the edge (as in Bojańczyk and Czerwiński~\cite[Section 3]{separation-toolbox}).
		
		We believe that it is possible to construct a transition strategy that does not depend on $G$ and $\tau$, 
		and that chooses a next transition basing only on priorities (i.e., without knowing which nodes are visited).
		The proof of existence of such a transition strategy would be more involved than the proof presented above, however.
		
		Nevertheless, $\Rr_{n,d}$ and $\Ss_{n,d}$ are not good for games, due to some words (not being in $\poseven_{n,d}$) that can be accepted by these automata, 
		but not in a deterministic way. 
		Interestingly, $\Rr_{n,d}$ accepts exactly $\limsupeven_{n,d}$, the set of winning plays (while $\Ll(\Ss_{n,d})$ is smaller).
	\end{remark}
	
\bibliography{bib}

\begin{thebibliography}{10}

\bibitem{priority-promotion}
Massimo Benerecetti, Daniele Dell'Erba, and Fabio Mogavero.
\newblock Solving parity games via priority promotion.
\newblock {\em Formal Methods in System Design}, 52(2):193--226, 2018.
\newblock \href {http://dx.doi.org/10.1007/s10703-018-0315-1}
  {\path{doi:10.1007/s10703-018-0315-1}}.

\bibitem{parity-to-safety-Walukiewicz}
Julien Bernet, David Janin, and Igor Walukiewicz.
\newblock Permissive strategies: From parity games to safety games.
\newblock {\em {ITA}}, 36(3):261--275, 2002.
\newblock \href {http://dx.doi.org/10.1051/ita:2002013}
  {\path{doi:10.1051/ita:2002013}}.

\bibitem{randomized-subexponential}
Henrik Bj{\"{o}}rklund and Sergei~G. Vorobyov.
\newblock A combinatorial strongly subexponential strategy improvement
  algorithm for mean payoff games.
\newblock {\em Discrete Applied Mathematics}, 155(2):210--229, 2007.
\newblock \href {http://dx.doi.org/10.1016/j.dam.2006.04.029}
  {\path{doi:10.1016/j.dam.2006.04.029}}.

\bibitem{separation-toolbox}
Mikołaj Boja\'nczyk and Wojciech Czerwi\'nski.
\newblock An automata toolbox, February 2018.
\newblock URL:
  \url{https://www.mimuw.edu.pl/~bojan/papers/toolbox-reduced-feb6.pdf}.

\bibitem{BCJLM97}
Anca Browne, Edmund~M. Clarke, Somesh Jha, David~E. Long, and Wilfredo~R.
  Marrero.
\newblock An improved algorithm for the evaluation of fixpoint expressions.
\newblock {\em Theor. Comput. Sci.}, 178(1-2):237--255, 1997.
\newblock \href {http://dx.doi.org/10.1016/S0304-3975(96)00228-9}
  {\path{doi:10.1016/S0304-3975(96)00228-9}}.

\bibitem{calude}
Cristian~S. Calude, Sanjay Jain, Bakhadyr Khoussainov, Wei Li, and Frank
  Stephan.
\newblock Deciding parity games in quasipolynomial time.
\newblock In Hamed Hatami, Pierre McKenzie, and Valerie King, editors, {\em
  Proceedings of the 49th Annual {ACM} {SIGACT} Symposium on Theory of
  Computing, {STOC} 2017, Montreal, QC, Canada, June 19-23, 2017}, pages
  252--263. {ACM}, 2017.
\newblock \href {http://dx.doi.org/10.1145/3055399.3055409}
  {\path{doi:10.1145/3055399.3055409}}.

\bibitem{universal-graphs}
Thomas Colcombet and Nathana{\"{e}}l Fijalkow.
\newblock Universal graphs and good for games automata: New tools for infinite
  duration games.
\newblock In Mikołaj Bojańczyk and Alex Simpson, editors, {\em Foundations of
  Software Science and Computation Structures - 22nd International Conference,
  {FOSSACS} 2019, Held as Part of the European Joint Conferences on Theory and
  Practice of Software, {ETAPS} 2019, Prague, Czech Republic, April 6-11, 2019,
  Proceedings}, volume 11425 of {\em Lecture Notes in Computer Science}, pages
  1--26. Springer, 2019.
\newblock \href {http://dx.doi.org/10.1007/978-3-030-17127-8_1}
  {\path{doi:10.1007/978-3-030-17127-8_1}}.

\bibitem{stochastic}
Anne Condon.
\newblock The complexity of stochastic games.
\newblock {\em Inf. Comput.}, 96(2):203--224, 1992.
\newblock \href {http://dx.doi.org/10.1016/0890-5401(92)90048-K}
  {\path{doi:10.1016/0890-5401(92)90048-K}}.

\bibitem{lower-bound}
Wojciech Czerwiński, Laure Daviaud, Nathana{\"{e}}l Fijalkow, Marcin
  Jurdziński, Ranko Lazić, and Paweł Parys.
\newblock Universal trees grow inside separating automata: Quasi-polynomial
  lower bounds for parity games.
\newblock In Timothy~M. Chan, editor, {\em Proceedings of the Thirtieth Annual
  {ACM-SIAM} Symposium on Discrete Algorithms, {SODA} 2019, San Diego,
  California, USA, January 6-9, 2019}, pages 2333--2349. {SIAM}, 2019.
\newblock \href {http://dx.doi.org/10.1137/1.9781611975482.142}
  {\path{doi:10.1137/1.9781611975482.142}}.

\bibitem{Daskalakis-Papadimitriou}
Constantinos Daskalakis and Christos~H. Papadimitriou.
\newblock Continuous local search.
\newblock In Dana Randall, editor, {\em Proceedings of the Twenty-Second Annual
  {ACM-SIAM} Symposium on Discrete Algorithms, {SODA} 2011, San Francisco,
  California, USA, January 23-25, 2011}, pages 790--804. {SIAM}, 2011.
\newblock \href {http://dx.doi.org/10.1137/1.9781611973082.62}
  {\path{doi:10.1137/1.9781611973082.62}}.

\bibitem{EJ91}
E.~Allen Emerson and Charanjit~S. Jutla.
\newblock Tree automata, mu-calculus and determinacy (extended abstract).
\newblock In {\em 32nd Annual Symposium on Foundations of Computer Science, San
  Juan, Puerto Rico, 1-4 October 1991}, pages 368--377. {IEEE} Computer
  Society, 1991.
\newblock \href {http://dx.doi.org/10.1109/SFCS.1991.185392}
  {\path{doi:10.1109/SFCS.1991.185392}}.

\bibitem{EJS01}
E.~Allen Emerson, Charanjit~S. Jutla, and A.~Prasad Sistla.
\newblock On model checking for the {\(\mathrm{\mu}\)}-calculus and its
  fragments.
\newblock {\em Theor. Comput. Sci.}, 258(1-2):491--522, 2001.
\newblock \href {http://dx.doi.org/10.1016/S0304-3975(00)00034-7}
  {\path{doi:10.1016/S0304-3975(00)00034-7}}.

\bibitem{FearnleyMDP}
John Fearnley.
\newblock Exponential lower bounds for policy iteration.
\newblock In Samson Abramsky, Cyril Gavoille, Claude Kirchner, Friedhelm {Meyer
  auf der Heide}, and Paul~G. Spirakis, editors, {\em Automata, Languages and
  Programming, 37th International Colloquium, {ICALP} 2010, Bordeaux, France,
  July 6-10, 2010, Proceedings, Part {II}}, volume 6199 of {\em Lecture Notes
  in Computer Science}, pages 551--562. Springer, 2010.
\newblock \href {http://dx.doi.org/10.1007/978-3-642-14162-1_46}
  {\path{doi:10.1007/978-3-642-14162-1_46}}.

\bibitem{Fearnley}
John Fearnley, Sanjay Jain, Sven Schewe, Frank Stephan, and Dominik Wojtczak.
\newblock An ordered approach to solving parity games in quasi polynomial time
  and quasi linear space.
\newblock In Hakan Erdogmus and Klaus Havelund, editors, {\em Proceedings of
  the 24th {ACM} {SIGSOFT} International {SPIN} Symposium on Model Checking of
  Software, Santa Barbara, CA, USA, July 10-14, 2017}, pages 112--121. {ACM},
  2017.
\newblock \href {http://dx.doi.org/10.1145/3092282.3092286}
  {\path{doi:10.1145/3092282.3092286}}.

\bibitem{parity-to-safety-PDA}
Wladimir Fridman and Martin Zimmermann.
\newblock Playing pushdown parity games in a hurry.
\newblock In Marco Faella and Aniello Murano, editors, {\em Proceedings Third
  International Symposium on Games, Automata, Logics and Formal Verification,
  GandALF 2012, Napoli, Italy, September 6-8, 2012.}, volume~96 of {\em
  {EPTCS}}, pages 183--196, 2012.
\newblock \href {http://dx.doi.org/10.4204/EPTCS.96.14}
  {\path{doi:10.4204/EPTCS.96.14}}.

\bibitem{FHZ-simplex}
Oliver Friedmann, Thomas~Dueholm Hansen, and Uri Zwick.
\newblock Subexponential lower bounds for randomized pivoting rules for the
  simplex algorithm.
\newblock In Lance Fortnow and Salil~P. Vadhan, editors, {\em Proceedings of
  the 43rd {ACM} Symposium on Theory of Computing, {STOC} 2011, San Jose, CA,
  USA, 6-8 June 2011}, pages 283--292. {ACM}, 2011.
\newblock \href {http://dx.doi.org/10.1145/1993636.1993675}
  {\path{doi:10.1145/1993636.1993675}}.

\bibitem{parity-short}
Hugo Gimbert and Rasmus Ibsen{-}Jensen.
\newblock A short proof of correctness of the quasi-polynomial time algorithm
  for parity games.
\newblock {\em CoRR}, abs/1702.01953, 2017.
\newblock \href {http://arxiv.org/abs/1702.01953} {\path{arXiv:1702.01953}}.

\bibitem{parity-to-safety-CPDA}
Matthew Hague, Roland Meyer, Sebastian Muskalla, and Martin Zimmermann.
\newblock Parity to safety in polynomial time for pushdown and collapsible
  pushdown systems.
\newblock In Igor Potapov, Paul~G. Spirakis, and James Worrell, editors, {\em
  43rd International Symposium on Mathematical Foundations of Computer Science,
  {MFCS} 2018, August 27-31, 2018, Liverpool, {UK}}, volume 117 of {\em
  LIPIcs}, pages 57:1--57:15. Schloss Dagstuhl - Leibniz-Zentrum fuer
  Informatik, 2018.
\newblock \href {http://dx.doi.org/10.4230/LIPIcs.MFCS.2018.57}
  {\path{doi:10.4230/LIPIcs.MFCS.2018.57}}.

\bibitem{good-for-games}
Thomas~A. Henzinger and Nir Piterman.
\newblock Solving games without determinization.
\newblock In Zolt{\'a}n {\'E}sik, editor, {\em Computer Science Logic}, volume
  4207 of {\em Lecture Notes in Computer Science}, pages 395--410, Berlin,
  Heidelberg, 2006. Springer Berlin Heidelberg.
\newblock \href {http://dx.doi.org/10.1007/11874683_26}
  {\path{doi:10.1007/11874683_26}}.

\bibitem{up-co-up}
Marcin Jurdziński.
\newblock Deciding the winner in parity games is in {UP} $\cap$ co-{UP}.
\newblock {\em Inf. Process. Lett.}, 68(3):119--124, 1998.
\newblock \href {http://dx.doi.org/10.1016/S0020-0190(98)00150-1}
  {\path{doi:10.1016/S0020-0190(98)00150-1}}.

\bibitem{old-progress-measure}
Marcin Jurdziński.
\newblock Small progress measures for solving parity games.
\newblock In Horst Reichel and Sophie Tison, editors, {\em {STACS} 2000, 17th
  Annual Symposium on Theoretical Aspects of Computer Science, Lille, France,
  February 2000, Proceedings}, volume 1770 of {\em Lecture Notes in Computer
  Science}, pages 290--301. Springer, 2000.
\newblock \href {http://dx.doi.org/10.1007/3-540-46541-3_24}
  {\path{doi:10.1007/3-540-46541-3_24}}.

\bibitem{small-progress-measure}
Marcin Jurdziński and Ranko Lazić.
\newblock Succinct progress measures for solving parity games.
\newblock In {\em 32nd Annual {ACM/IEEE} Symposium on Logic in Computer
  Science, {LICS} 2017, Reykjavik, Iceland, June 20-23, 2017}, pages 1--9.
  {IEEE} Computer Society, 2017.
\newblock \href {http://dx.doi.org/10.1109/LICS.2017.8005092}
  {\path{doi:10.1109/LICS.2017.8005092}}.

\bibitem{subexponential}
Marcin Jurdziński, Mike Paterson, and Uri Zwick.
\newblock A deterministic subexponential algorithm for solving parity games.
\newblock {\em {SIAM} J. Comput.}, 38(4):1519--1532, 2008.
\newblock \href {http://dx.doi.org/10.1137/070686652}
  {\path{doi:10.1137/070686652}}.

\bibitem{parity-excursion}
Bakhadyr Khoussainov.
\newblock A brief excursion to parity games.
\newblock In Mizuho Hoshi and Shinnosuke Seki, editors, {\em Developments in
  Language Theory - 22nd International Conference, {DLT} 2018, Tokyo, Japan,
  September 10-14, 2018, Proceedings}, volume 11088 of {\em Lecture Notes in
  Computer Science}, pages 24--35. Springer, 2018.
\newblock \href {http://dx.doi.org/10.1007/978-3-319-98654-8_3}
  {\path{doi:10.1007/978-3-319-98654-8_3}}.

\bibitem{Lehtinen}
Karoliina Lehtinen.
\newblock A modal {\(\mu\)} perspective on solving parity games in
  quasi-polynomial time.
\newblock In Anuj Dawar and Erich Gr{\"{a}}del, editors, {\em Proceedings of
  the 33rd Annual {ACM/IEEE} Symposium on Logic in Computer Science, {LICS}
  2018, Oxford, UK, July 09-12, 2018}, pages 639--648. {ACM}, 2018.
\newblock \href {http://dx.doi.org/10.1145/3209108.3209115}
  {\path{doi:10.1145/3209108.3209115}}.

\bibitem{Mos91}
Andrzej~W. Mostowski.
\newblock Games with forbidden positions.
\newblock Technical Report~78, Uniwersytet Gda\'nski, 1991.

\bibitem{ZielonkaParys}
Paweł Parys.
\newblock Parity games: Zielonka's algorithm in quasi-polynomial time.
\newblock In Peter Rossmanith, Pinar Heggernes, and Joost{-}Pieter Katoen,
  editors, {\em 44th International Symposium on Mathematical Foundations of
  Computer Science, {MFCS} 2019, August 26-30, 2019, Aachen, Germany}, volume
  138 of {\em LIPIcs}, pages 10:1--10:13. Schloss Dagstuhl - Leibniz-Zentrum
  f{\"{u}}r Informatik, 2019.
\newblock \href {http://dx.doi.org/10.4230/LIPIcs.MFCS.2019.10}
  {\path{doi:10.4230/LIPIcs.MFCS.2019.10}}.

\bibitem{RabinBook}
Michael~Oser Rabin.
\newblock {\em Automata on Infinite Objects and Church's Problem}.
\newblock American Mathematical Society, Boston, MA, USA, 1972.

\bibitem{Schewe-big-steps}
Sven Schewe.
\newblock Solving parity games in big steps.
\newblock {\em J. Comput. Syst. Sci.}, 84:243--262, 2017.
\newblock \href {http://dx.doi.org/10.1016/j.jcss.2016.10.002}
  {\path{doi:10.1016/j.jcss.2016.10.002}}.

\bibitem{Seidl96}
Helmut Seidl.
\newblock Fast and simple nested fixpoints.
\newblock {\em Inf. Process. Lett.}, 59(6):303--308, 1996.
\newblock \href {http://dx.doi.org/10.1016/0020-0190(96)00130-5}
  {\path{doi:10.1016/0020-0190(96)00130-5}}.

\bibitem{strategy-improvement}
Jens V{\"{o}}ge and Marcin Jurdziński.
\newblock A discrete strategy improvement algorithm for solving parity games.
\newblock In E.~Allen Emerson and A.~Prasad Sistla, editors, {\em Computer
  Aided Verification, 12th International Conference, {CAV} 2000, Chicago, IL,
  USA, July 15-19, 2000, Proceedings}, volume 1855 of {\em Lecture Notes in
  Computer Science}, pages 202--215. Springer, 2000.
\newblock \href {http://dx.doi.org/10.1007/10722167_18}
  {\path{doi:10.1007/10722167_18}}.

\bibitem{Zielonka}
Wiesław Zielonka.
\newblock Infinite games on finitely coloured graphs with applications to
  automata on infinite trees.
\newblock {\em Theor. Comput. Sci.}, 200(1-2):135--183, 1998.
\newblock \href {http://dx.doi.org/10.1016/S0304-3975(98)00009-7}
  {\path{doi:10.1016/S0304-3975(98)00009-7}}.

\bibitem{mean-payoff}
Uri Zwick and Mike Paterson.
\newblock The complexity of mean payoff games on graphs.
\newblock {\em Theor. Comput. Sci.}, 158(1{\&}2):343--359, 1996.
\newblock \href {http://dx.doi.org/10.1016/0304-3975(95)00188-3}
  {\path{doi:10.1016/0304-3975(95)00188-3}}.

\end{thebibliography}

\appendix
\section{Appendix}

\newcommand{\plays}{\mathrm{Plays}}

In Remark~\ref{rem:1} we claim that $\Rr_{n,d}$ and $\Ss_{n,d}$ are not good for games, that $\Ll(\Rr_{n,d})=\limsupeven_{n,d}$, and that $\Ll(\Ss_{n,d})\subsetneq\limsupeven_{n,d}$.
We prove these facts below.

\begin{lemma}\label{lem:A1}
	It holds that $\Ll(\Rr_{n,d})=\limsupeven_{n,d}$.
\end{lemma}

\begin{proof}
	Because $\Rr_{n,d}$ is a strong PG separator, we already know that it rejects all words from $\limsupodd_{n,d}$.
	It remains to prove that it accepts all words from $\limsupeven_{n,d}$.
	Consider thus a word $w\in\limsupeven_{n,d}$.
	Let $p$ be the greatest priority that appears in $w$ infinitely often; $p$ is even.
	We consider a run $\rho$ of $\Rr_{n,d}$ that performs a reset of the register $1$ while reading an edge with priority $p$, and a non-reset transition while reading any other edge.
	Let us see that $\rho$ is accepting.
	From some moment on, there are no more priorities higher than $p$ in $w$.
	Then, after reading the priority $p$, the register $1$ is reset to $1$, and later it is never updated to any value higher than $p$.
	Thus, whenever the priority $p$ is read, we set the register $1$ to $p$, and then we reset it---it is an even reset.
	We thus have infinitely many even resets of the register $1$, only finitely many odd resets of this register, and no resets of any higher register; the run is accepting.
\end{proof}

\begin{lemma}
	For $d\geq2\cdot\rn(n)+2$, the automaton $\Rr_{n,d}$ is not good for games.
\end{lemma}

\begin{proof}
	Consider some transition strategy $\sigma$ for $\Rr_{n,d}$.
	We have to prove that $\sigma$ is not winning for the whole $\Ll(\Rr_{n,d})$ 
	(having this for every transition strategy $\sigma$ implies that $\Rr_{n,d}$ is not good for games).
	
	To this end, we construct a play $\xi$ by induction.
	The first edge of $\xi$ is $(1,2,1)$.
	Supposing that some prefix of $\xi$ is already constructed, we run $\Rr_{n,d}$ on this prefix, following the transition strategy $\sigma$.
	If the last transition of such a partial run is
	\begin{itemize}
	\item	a non-reset transition---we append $(1,2,1)$ to $\xi$;
	\item	an odd reset of a register $k$---we append $(1,2k+2,1)$ to $\xi$;
	\item	an even reset of a register $k$---we append $(1,2k+1,1)$ to $\xi$.
	\end{itemize}
	By the assumption $d\geq2\cdot \rn(n)+2$, all edges appearing in the above definition belong to $\Sigma_{n,d}$.

	Suppose that $\rho$ is accepting.
	Then, for some number $k$, in $\rho$ there are infinitely many even resets of the register $k$, and only finitely many odd resets of the register $k$ and resets of higher registers.
	Thus, the highest priority that is appended to $\xi$ infinitely often is $2k+1$ (by the definition of $\xi$).
	This means that $\xi\in\limsupodd_{n,d}$, so $\xi$ is rejected by $\Rr_{n,d}$, which is a contradiction.
	
	We thus have that $\rho$ is rejecting.
	By the definition of $\Rr_{n,d}$ this means that in $\rho$ either
	\begin{itemize}
	\item	from some moment on, there are only non-reset transitions, or
	\item	for some register $k$, in $\rho$ there are infinitely many odd resets of the register $k$, and only finitely many resets of higher registers.
	\end{itemize}
	Then, the highest priority appended to $\xi$ infinitely often is even ($2$ in the former case, and $2k+2$ in the latter case), so $\xi\in\limsupeven_{n,d}$,
	that is, $\xi\in\Ll(\Rr_{n,d})$ (by Lemma~\ref{lem:A1}).
	This means that $\sigma$ is not winning for $\Ll(\Rr_{n,d})$, as we wanted to prove.
\end{proof}

\begin{lemma}\label{lem:A3}
	For $n\geq 2$ and $d\geq 6$, the automaton $\Ss_{n,d}$ is not good for games.
\end{lemma}

\begin{proof}
	Denote $r=\rn(n)$.
	The assumption $n\geq2$ can be rewritten as $r\geq 2$; in the proof below we indeed use the fact that there are at least two registers.
	Denote also $s_I=(\seq{1,1,\dots,1},\seq{n,n,\dots,n})$ (this is the initial state of $\Ss_{n,d}$) and $s_2=(\seq{4,1,\dots,1},\seq{n,1,\dots,1})$.
	Recall that $\rej$ denotes the rejecting state of $\Ss_{n,d}$.
	For $i=0,1,2,3$, respectively, let $k_i$ be the largest number such that there exists a partial run of $\Ss_{n,d}$ that 
	\begin{itemize}
	\item	starts in $s_I$, ends in $s_2$, and reads the word $(1,4,1)(1,1,1)^{k_0}$;
	\item	starts in $s_2$, avoids $\rej$, and reads the word $(1,1,1)^{k_1}$;
	\item	starts in $s_2$, avoids $\rej$, and reads the word $(1,2,1)(1,3,1)^{k_2}$;
	\item	starts in $s_2$, avoids $\rej$, and reads the word $(1,2,1)(1,5,1)^{k_3}$.
	\end{itemize}
	Notice that there exists a partial run from $s_I$ to $s_2$ that reads $(1,4,1)(1,1,1)^k$ for some $k$ (thus $k_0$ is well defined):
	for example, the state $s_2$ is reached if we first reset $r-1$ times the register $r-1$ having value $4$ (even resets, changing the state to $(\seq{4,1,\dots,1},\seq{n,n,\dots,n})$),
	then, for $j=r-1,r-2,\dots,1$, we reset $n-1$ times the register $j$ having value $1$ (odd resets), and finally we can perform $n-1$ non-reset transitions.
	Likewise, if while reading $(1,2,1)$ from $s_2$ we reset some of the registers (this is an even reset: the registers have value $2$ or $4$), 
	we do not reach $\rej$, so $k_2$ and $k_3$ are well defined.
	Notice also that all $k_i$ are finite.
	Indeed, suppose for example that there is a partial run from $s_I$ to $s_2$ (i.e., not visiting $\rej$) that reads $(1,4,1)(1,1,1)^k$ for $k$ larger than the number of states of $\Ss_{n,d}$.
	Then, it contains a cycle.
	We can repeat this cycle forever, obtaining an infinite run reading $(1,4,1)(1,1,1)^\omega$.
	This run is accepting (it does not visit $\mathsf{rej}$),
	but the word is in $\limsupodd_{n,d}$, which is impossible: $\Ss_{n,d}$ rejects all words from $\limsupodd_{n,d}$.
	This shows that $k_0$ is finite, and similarly we can prove that $k_1,k_2,k_3$ are finite.
	
	Consider now three infinite plays:
	\begin{align*}
		w_1&=(1,4,1)(1,1,1)^{k_0}(1,1,1)^{k_1}(1,6,1)^\omega\,,\\
		w_2&=(1,4,1)(1,1,1)^{k_0}(1,2,1)(1,3,1)^{k_2}(1,6,1)^\omega\,,\\
		w_3&=(1,4,1)(1,1,1)^{k_0}(1,2,1)(1,5,1)^{k_3}(1,6,1)^\omega\,.
	\end{align*}
	
	Observe that the words can be accepted by $\Ss_{n,d}$.
	Indeed, by the definition of the numbers $k_i$, there is a partial run from $s_I$ that avoids $\rej$ and reads the finite prefix of $w_j$ (for $j=1,2,3$) before $(1,6,1)^\omega$.
	Then, while reading $(1,6,1)^\omega$, we can perform an even reset of the register $r$ (having value $6$); even resets do not lead to $\rej$, so the run is accepting.

	Next, we prove that there is no single transition strategy that allows to accept simultaneously $w_1$, $w_2$, and $w_3$.
	To this end, suppose that $\rho_1$, $\rho_2$, $\rho_3$ are accepting runs reading respectively $w_1$, $w_2$, $w_3$, and following the same transition strategy.

	Let $s_2'$ be the state reached after reading the prefix $(1,4,1)(1,1,1)^{k_0}$ (this state is the same in the three runs).
	As the runs are accepting, we have $s_2'\neq\rej$.
	Because the priorities read so far are only $1$ and $4$, the registers of $s_2'$ can contain only values $1$ or $4$: 
	say that there is $1$ in the lower $l$ registers, and $4$ in the remaining $r-l$ registers.
	
	Suppose first that $l=r$, that is, that $s_2'=(\seq{1,1,\dots,1},\seq{c_r,c_{r-1},\dots,c_0})$.
	In this situation there is a partial run from $s_2=(\seq{4,1,\dots,1},\seq{n,1,\dots,1})$ to $s_2'$ reading $(1,1,1)^{k_1'}$ for some $k_1'>0$.
	Indeed, we first reset the register $r$ having value $4$, which leads to $(\seq{1,1,\dots,1},\seq{n,n,\dots,n})$;
	then, for $j=r,r-1,\dots,1$, we reset $n-c_j$ times the register $j$ having value $1$ (odd resets);
	finally, we perform $n-c_0$ non-reset transitions.
	It follows that there is a partial run from $s_2$ that avoids $\rej$ and reads $(1,1,1)^{k_1'+k_1}$: we first reach $s_2'$, and then we follow $\rho_1$.
	This contradicts the maximality of $k_1$; it was impossible that $l=r$.
	
	We thus have $l\leq r-1$.
	This is only possible if the register $r$ contained the value $4$ in all states of the considered prefix of the run (except the initial state).
	Thus there were no odd resets of the register $r$ so far; the counter number $r$ in $s_2'$ has value $n$.
	We now observe that from $s_2'$ we can reach $s_2$ reading $(1,1,1)^{k_0'}$ for some $k_0'$.
	Indeed, if $l=r-1$, we have $s_2'=(\seq{4,1,\dots,1},\seq{n,c_{r-1},\dots,c_0})$; by performing an appropriate number of odd resets of registers $r-1,\dots,1$, and of non-reset transitions,
	we can change the counters to $\seq{n,1,\dots,1}$.
	If $l<r-1$, we first reset $r-1-l$ times the register $r-1$ having value $4$ (even resets), which changes the state to $(\seq{4,1,\dots,1},\seq{n,n,\dots,n})$), and then we continue as above.
	Notice now that a composition of the considered prefix of the run and of this partial run is a partial run from from $s_I$ to $s_2$ that reads $(1,4,1)(1,1,1)^{k_0+k_0'}$.
	By maximality of $k_0$ we have that $k_0'=0$, which is only possible if $s_2'=s_2$.
	
	We already know that after reading $(1,4,1)(1,1,1)^{k_0}$ the runs $\rho_i$ reach the state $s_2$.
	Consider now the transition of $\rho_2$ and $\rho_3$ reading $(1,2,1)$; in both runs this is the same transition.
	A non-reset transition from $s_2$ would lead to $\rej$ (the last counter in $s_2$ contains $1$), so this is a reset transition.
	One case is that this transition resets the register $r$, and leads to $s_3=(\seq{2,2,\dots,2,1},\seq{n,n,n,\dots,n})$.
	We can instead reset the register $r-1$ while reading $(1,2,1)$, which leads to $s_3'=(\seq{4,2,\dots,2,1},\seq{n,1,n,\dots,n})$,
	and then reset the register $r$ while reading $(1,3,1)$, which leads to $s_3''=(\seq{3,3,\dots,3,1},\seq{n,n,n,\dots,n})$.
	We know that there is a partial run from $s_3$ that reads $(1,3,1)^{k_2}$ and avoids $\rej$ (a fragment of $\rho_2$).
	But the same partial run can be executed from $s_3''$ (an update by the priority $3$ results in the same, no matter whether we start from $s_3$ or from $s_3''$).
	Thus, by going through $s_3'$ and $s_3''$ we obtain a partial run from $s_2$ that reads $(1,2,1)(1,3,1)^{1+k_2}$ and avoids $\rej$.
	This contradicts the maximality of $k_2$.
	
	The remaining case is that the transition of $\rho_2$ and $\rho_3$ reading $(1,2,1)$ resets some register other than $r$.
	The transition leads to a state $s_4=(\seq{4,2,\dots,2,1},\seq{n,1,c_{r-2},\dots,c_1,n})$ 
	(where the counter number $r-1$ remains $1$, and the counters $c_{r-2},\dots,c_1$ are either $1$ or $n$).
	We can instead reset the register $r$ while reading $(1,2,1)$, which leads to $s_3$ as above,
	and then, while reading $(1,5,1)^{k_3'}$, we can reset $n-1$ times the register $r-1$ (odd resets, leading to $(\seq{5,5,\dots,5,1},\seq{n,1,n,\dots,n})$),
	and reset $n-1$ times appropriate registers among $r-2,\dots,1$, so that the state becomes $s_4'=(\seq{5,5,\dots,5,1},\seq{n,1,c_{r-2},\dots,c_1,n})$.
	Notice that $k_3'\geq n-1>0$.
	We know that there is a partial run from $s_4$ that reads $(1,5,1)^{k_3}$ and avoids $\rej$ (a fragment of $\rho_3$).
	But the same partial run can be executed from $s_4'$ (an update by the priority $5$ results in the same, no meter whether we start from $s_4$ or from $s_4'$).
	Thus, by going through $s_3$ and $s_4'$ we obtain a partial run from $s_2$ that reads $(1,2,1)(1,5,1)^{k_3'+k_3}$ and avoids $\rej$.
	This contradicts the maximality of $k_3$.

	We thus obtain that $w_1,w_2,w_3\in\Ll(\Ss_{n,d})$ cannot be accepted while following a common transition strategy.
	In other words, no transition strategy is winning for the whole $\Ll(\Ss_{n,d})$; the automaton is not good for games.
\end{proof}

\begin{lemma}
	For $d\geq 2$ we have that $\Ll(\Ss_{n,d})\subsetneq\limsupeven_{n,d}$.
\end{lemma}

\begin{proof}
	We already know that $\Ss_{n,d}$ rejects all words from $\limsupodd_{n,d}$; it is thus enough to prove that $\limsupeven_{n,d}\setminus\Ll(\Ss_{n,d})\neq\emptyset$.
	Let $k$ be the largest number such that there exists a partial run of $\Ss_{n,d}$ that starts in the initial state, avoid $\rej$, and reads $(1,1,1)^k$.
	As in the proof of Lemma~\ref{lem:A3}, the number $k$ is finite (not larger than the number of states of $\Ss_{n,d}$).
	Consider the word $(1,1,1)^{k+1}(1,2,1)^\omega\in\limsupeven_{n,d}$.
	Because at the beginning we have $(1,1,1)$ repeated $k+1$ times, every run reading this word is rejecting 
	(otherwise existence of the partial run reading $(1,1,1)^{k+1}$ and avoiding $\rej$ would contradict the maximality of $k$).
	Thus, this word belongs to $\limsupeven_{n,d}\setminus\Ll(\Ss_{n,d})$.
\end{proof}

\end{document}